\documentclass[pra,12pt,tightenlines]{revtex4}
\usepackage{amssymb}
\usepackage{graphicx}
\usepackage{amsmath}
\usepackage{amsbsy}
\usepackage{amsthm}
\usepackage{bbm}
\usepackage{bm}
\theoremstyle{plain}

\newtheorem{lem}{Lemma}
\newtheorem{cor}{Corollary}
\newtheorem{prop}{Proposition}
\theoremstyle{definition}
\newtheorem{rem}{Remark}
\newtheorem{defin}{Definition}
\newcommand{\id}{\mathbbm{1}}
\newcommand{\tr}{\textnormal{Tr}}

\begin{document}

\title{Characterizing entanglement with\\ geometric entanglement witnesses}
\author{Philipp Krammer} \email{philipp.krammer[at]univie.ac.at}
\affiliation{Faculty of Physics, University of Vienna,
Boltzmanngasse 5, A-1090 Vienna, Austria}

\begin{abstract}
We show how to detect entangled, bound entangled, and separable
bipartite quantum states of arbitrary dimension and mixedness using
geometric entanglement witnesses. These witnesses are constructed
using properties of the Hilbert-Schmidt geometry and can be shifted
along parameterized lines. The involved conditions are simplified
using Bloch decompositions of operators and states. As an example we
determine the three different types of states for a family of
two-qutrit states that is part of the ``magic simplex'', i.e. the
set of Bell-state mixtures of arbitrary dimension.
\end{abstract}


\maketitle

\section{Introduction}

Entanglement is a fascinating curiosity of quantum physics that
distinguishes it considerably from classical concepts
\cite{schroedinger35}. On the one hand it implicates surprising
philosophical aspects such as the incompatibility of local realistic
theories with quantum physics \cite{bell64, clauser69}, on the other
hand it can be successfully implied in quantum information and
quantum communication tasks to improve quantum protocols with
respect to classical ones (for an overview see, e.g., Refs
\cite{nielsen00, bouwmeester00, bertlmann02a}).

It is still an open mathematical problem to determine whether a
quantum state is entangled or not, there is no operational procedure
for an arbitrary state, although there are many useful criteria that
allow a detection of entanglement in many cases \cite{bruss02,
horodecki01, horodecki07}. For pure states and lower dimensional
bipartite systems (e.g., two qubits) the problem is solved, since
there exist applicable necessary and sufficient conditions for
separability (i.e. non-entanglement) \cite{horodecki96}.
Additionally, entangled states can be classified according to their
\emph{distillability}: A distillable state can be ``distilled'' to a
(nearly) maximally entangled state via statistical local operations
and classical communication (SLOCC). States that are not distillable
are called \emph{bound entangled}, whereas distillable states are
called \emph{free entangled} \cite{horodecki97a, horodecki98,
rains99, horodecki07}. Examples of bound entangled states and
construction procedures can be found in Refs.~\cite{bennett98,
divincenzo03, horodecki99c, bruss00, hyllus04}.

In this article we want to present a method to classify entanglement
with entanglement witnesses, which provide a well-established tool
to detect entanglement. Our approach is based on two concepts: We
construct witnesses in a geometrically intuitive way and use Bloch
decompositions of operators and states to simplify the mathematical
application. We show how to use geometric entanglement witnesses to
detect more entangled states than the PPT criterion and how to
identify separable states by using optimal entanglement witnesses.
This is attained by shifting the witnesses along parameterized lines
of states. Two main methods are explained in detail: the
\emph{outside-in shift} and the \emph{inside-out shift}. The
\emph{outside-in shift} is used for detecting more (bound) entangled
states, whereas with the \emph{inside-out shift} we construct the
shape of the set of separable states.

The paper is organized as follows: In Sec.~\ref{secbasics} we give
an overview of the mathematical basics of entanglement theory that
will be used throughout this article. We present a formulation of
the entanglement witness criterion in terms of Bloch decompositions
in Sec.~\ref{secbloch} and explain a geometric method to construct
entanglement witnesses in  Sec.~\ref{secgew}, where we also
introduce the two shift methods. The results are applied to a family
of states that are part of a special simplex (called ``magic
simplex'') in the Hilbert-Schmidt space of two qutrits in
Sec.~\ref{secexample}, where illustrative geometric pictures are
obtained identifying regions of entangled, bound entangled, and
separable states.

\section{Basic concepts of entanglement theory} \label{secbasics}

We consider a bipartite Hilbert-Schmidt space ${\cal A}:= {\cal A}_A
\otimes {\cal A}_B$ on a discrete finite dimensional Hilbert space
${\cal H}:= {\cal H}_A \otimes {\cal H}_B$ of dimension $d_A \times
d_B$, $D:= d_A d_B$. Vector states $| \psi \rangle$ are elements of
${\cal H}$ whereas density operators (called ``density matrices'' or
just ``states'') are elements of ${\cal A}$, a Hilbert space of
operators on ${\cal H}$, with a scalar product
\begin{equation}
    \langle A, B \rangle = \tr A^\dag B , \quad \mbox{with }
    A, \, B \in {\cal A} \,.
\end{equation}
States $\rho \in {\cal A}$ are defined by the properties
\begin{equation}
    \rho^\dag = \rho, \ \tr \rho = 1, \ \rho \geq 0 \,,
\end{equation}
i. e. they are Hermitian, have trace one, and are positive
semidefinite, i.e. have nonnegative eigenvalues. For the sake of
simplicity we will often drop the term ``semidefinite'' in the text
and write ``positive'' only, meaning positive semidefinite. All
states $\rho$ satisfy the inequality $\tr \rho^2 \leq 1$, and are
classified as pure states ($\tr \rho^2 = 1$) or mixed states ($\tr
\rho^2 < 1$).

The Hilbert-Schmidt distance of two states $\rho_1$ and $\rho_2$ is
given by
\begin{equation} \label{hsdistance}
    d( \rho_1, \rho_2) := \| \rho_1 - \rho_2 \| := \langle \rho_1 -
    \rho_2, \rho_1 - \rho_2 \rangle ^{1/2} \,.
\end{equation}

Density operators can be represented as a matrix using an
orthonormal basis of ${\cal H}$. The standard basis representation
is
\begin{equation}
    \rho_{i j, k l} := \langle i j | \rho | k l \rangle
\end{equation}
where $| i j \rangle := | i \rangle \otimes | j \rangle$ and
$\left\{ | i j \rangle \right\}$ is the standard product basis of a
$d_A \times d_B$ dimensional system.

A state $\rho$ is called \emph{entangled} if it cannot be written as
a convex combination of product states \cite{werner89},
\begin{equation} \label{sepstate}
    \rho \neq \sigma = \sum_i p_i \rho_A^i \otimes \rho_B^i \,,
\end{equation}
where $\rho_A^i$ and $\rho_B^i$ are states of the two subsystems,
usually called ``Alice'' and ``Bob'' following the convention of
quantum communication. States $\sigma$ \eqref{sepstate} are called
\emph{separable}, they are not entangled and contain classical
correlations only. There is still no operational method (i.e., no
``recipe'') to decide for a given bipartite state of arbitrary
dimension wether it it can be written as the convex combination
\eqref{sepstate} or not. This problem is known as the separability
problem. However, there exist several criteria which help to find
the entanglement properties of states in special cases, for an
overview see, e.g., Refs.~\cite{bruss02, horodecki07, guehne08}

A criterion that is necessary and sufficient, but not operational,
is the entanglement witness criterion (EWC) \cite{horodecki96,
terhal00, terhal02, bertlmann02, bertlmann05}. It says that a state $\rho$
is entangled if and only if it can be ``witnessed'' by some
Hermitian operator $A$, for which
\begin{eqnarray}
    \left\langle \rho,A \right\rangle \;=\; \tr \, \rho A
    & \;<\; & 0 \,, \label{defentwitent} \\
    \left\langle \sigma,A \right\rangle = \tr \, \sigma A & \;\geq\; & 0 \qquad
    \forall \,\sigma \in {\cal S} \,, \label{defentwitsep}
\end{eqnarray}
where ${\cal S}$ denotes the convex and compact set of all separable
states, and $A$ is called an \emph{entanglement witness}. We call
inequality \eqref{defentwitent} ``entanglement condition'' and
inequality \eqref{defentwitsep} ``separability condition''. If there
exists a separable state $\tilde{\sigma}$ for which $\tr
\tilde{\sigma} A = 0$, then $A$ is called an \emph{optimal}
entanglement witness. It is closest to the set of separable states
and thus detects more entangled states than non-optimal witnesses.

The EWC is a consequence of the Hahn-Banach theorem of functional
analysis. It geometrically corresponds to the fact that an element
of a Banach space can always be separated by a hyperplane from a
convex and compact subset that does not contain the element (see,
e.g., Ref.~\cite{reed72} and Refs.~\cite{bruss02, bertlmann05} for
illustrations). Although it is intuitive and simple, the EWC is not
easy to implement given an arbitrary state $\rho$, since in general
it is difficult to find a suitable witness that satisfies
Eq.~\eqref{defentwitsep}, and even more difficult to state that
there does not exist any witness for this state, which would imply
separability. Nevertheless the criterion plays an important role in
the theoretical understanding of entanglement, and has the advantage
that a witness $A$ corresponds to a physical observable that can be
implemented in experiments. It therefore allows a detection of
entanglement without performing a full tomography of the state
\cite{guehne02, guehne03, barbieri03, altepeter05, skwara07}.

An operational criterion that is a necessary condition for
separability is the positive partial transpose (PPT) criterion
\cite{peres96}. It simply says that a separable state $\sigma$ stays
positive under partial transposition,
\begin{equation} \label{PPT}
    \sigma^\Gamma := (\id \otimes T) \sigma \geq 0 \,,
\end{equation}
where the partial transpose is a transposition $T$ with respect to
Bob's system only, $\rho_{i j, k l}^\Gamma := \rho_{i l, k j}$. As a
proof of Eq.~\eqref{PPT} one just has to recognize that if we apply
the partial transposition to a separable state \eqref{sepstate} the
transposition is performed on $\rho_B^i$ only, which does not change
the positivity of the state, and thus the whole separable state
stays positive. Therefore, if a state $\rho$ violates the criterion,
i.e. it is no longer positive under partial transposition, it has to
be an entangled state. We call a state that is positive under
partial transposition PPT, and a state that is not NPT. For
dimensions $2 \times 2$ and $2 \times 3$ it can be shown that the
criterion is necessary and sufficient \cite{horodecki96}, i.e. any
entangled state has to be NPT. In higher dimensions, however, there
exist entangled states that are PPT, it can be proven that such
states are bound entangled. Note that the reverse is not necessarily
true, it is still an open question if all NPT states are free
entangled, although there are strong implications that NPT bound
entangled states exist \cite{duer00, divincenzo00}.

In order to find PPT entangled states, one has to employ criteria
that are not equal to or weaker than the PPT criterion. In principle
the EWC is strongest since it detects the entanglement of all
entangled states, but it is more cumbersome to apply.

Another useful criterion is the realignment criterion (or cross norm
criterion) \cite{rudolph00, rudolph02, rudolph03, chen03}, which
again is a necessary condition for separability. It states that for
any separable state the sum of the singular values $s_i$ of a
\emph{realigned} density matrix $\sigma_R$ has to be smaller than or
equal to one,
\begin{equation} \label{realign}
    \sum_i s_i = \tr \sqrt{\sigma^\dag_R \sigma_R} \leq 1 \,,
\end{equation}
where $\left( \rho_{ij,kl} \right)_R := \rho_{ik,jl}$. The
realignment criterion is neither weaker nor stronger than the PPT
criterion, meaning that it detects some entangled states that the
PPT criterion does not, and vice versa. Thus an application of both
criteria is easy to perform and allows a detection of many entangled
states, both free and bound entangled, but they do not constitute a
necessary and sufficient criterion together, since there exist PPT
entangled states that are not detected by the realignment criterion
\cite{chen03}.

\section{Bloch decompositions and entanglement witnesses}
\label{secbloch}

Bloch decompositions are a convenient way to handle calculations in
high dimensional systems, since the usage of large matrices can be
avoided (for an overview see \cite{bertlmann08a} and references
therein). There also exist computable separability criteria based on
Bloch decompositions of states \cite{pittenger02, devicente07}.

Let us consider the Hilbert-Schmidt space for a one-particle state
first, for example the space of Alice's subsystem ${\cal A}_A$ on
${\cal H}_A$ of dimension $d_A:=d$ (all considerations are
equivalent for Bob). Since the Hilbert-Schmidt space ${\cal A}_A$ is
a vector space of operators, one can decompose any element of ${\cal
A}_A$ into a linear combination of operators that form an orthogonal
basis of the Hilbert-Schmidt space. Let us identify such a basis of
$d^2$ operators with $\left\{ \id, A_i \right\}$, $i=1, \ldots,
d^2-1$. The operators $A_i$ are traceless, $\tr A_i = 0$ and satisfy
the orthogonality condition
\begin{equation} \label{mbasisorth}
    \tr A_i A_j = N_A \delta_{ij}, \quad N_A \in \mathbbm{R} \,.
\end{equation}
A one-particle qudit state can then be decomposed into the operator
basis as (for example a state $\rho_A$ for Alice's subsystem of
dimension d)
\begin{equation} \label{blochonepart}
    \rho_A = \frac{1}{d} \left( \id + \sqrt{\frac{d(d-1)}{N_A}} \sum_{i=1}^{d^2-1}
    n_i A_i \right), \quad n_i \in \mathbbm{C} \,,
\end{equation}
where $|\vec{n}|^2 = \sum_i n_i^*n_i \leq 1$. The coefficient vector
$\vec{n}$ is called \emph{Bloch vector}, it uniquely characterizes
the state. The constant $\sqrt{d(d-1)/N_A}$ results from the inequality
$\tr \rho^2 \leq 1$. The state is pure if and only if $|\vec{n}|^2 =
1$. In general not all arbitrary vectors $\vec{n}$ are Bloch
vectors, i.e. they do not necessarily imply $\rho \geq 0$, see
Remark~\ref{remcorbvew}.

A bipartite product state $ \sigma_p := \rho_A \otimes \rho_B$ on
${\cal H}$ of dimension $d_A \times d_B$ can be written as (where
Bob's orthogonal basis is $\{\id, B_j \}$)
\begin{align} \label{prodbloch}
    \sigma_p = \ & \frac{1}{d^2} \Big( \id_{\rm{d_A}} \otimes
        \id_{\rm{d_B}}
        + \sum_{i=0}^{d^2_A - 1} f_A \, n_i \, A_i \otimes \id_{\rm{d_B}}
        + \sum_{j=0}^{d^2_B - 1} f_B \, m_j \,\id_{\rm{d_A}} \otimes B_j
        + \sum_{i,j} f_A f_B \, n_i m_j A_i \otimes B_j \Big) \,,
        \nonumber\\
    & n_{nm}, m_{lk} \in \mathbbm{C}\,, \ \left| \vec{n} \right|
    \leq 1\,, \ \left| \vec{m} \right| \leq 1 \,, \quad f_A := \sqrt{\frac{d_A(d_A-1)}{N_A}}, \ f_B :=
        \sqrt{\frac{d_B(d_B-1)}{N_B}} \,,
\end{align}
where the state is pure if and only if $\left| \vec{n} \right| =
\left| \vec{m} \right| = 1$.

Any operator $O \in {\cal A}$ can be decomposed as
\begin{align} \label{opbloch}
    O = & \, e \, \id_{\rm{d_A}} \otimes
        \id_{\rm{d_B}}
        + \sum_{i=0}^{d^2_A - 1} \, a_i \, A_i \otimes \id_{\rm{d_B}}
        + \sum_{j=0}^{d^2_B - 1} \, b_j \,\id_{\rm{d_A}} \otimes B_j
        + \sum_{i,j} \, c_{ij} A_i \otimes B_j \,,
        \nonumber\\
    & e, a_i, b_i, c_{ij} \in \mathbbm{C} \,.
\end{align}
For a given operator $O$ and same dimensions of the subsystems
$d_A=d_B=:d$, one can always find an orthogonal basis in which the
coefficient matrix $C^{cor} = (c_{ij})$, called \emph{correlation
coefficient matrix}, is diagonal: Given an operator decomposition
\eqref{opbloch}, we have to perform a singular value decomposition
of $C$,
\begin{equation}
    S = U C^{cor} V^\dag \,,
\end{equation}
where $U$ and $V$ are unitary matrices with entries $u_{ij}$ and
$v_{ij}$ and S is the resulting diagonal matrix with the $d^2$
diagonal real positive singular values $s_i$ of $C^{cor}$ as
diagonal entries. The new basis operators $D_i^A$ and $D_i^B$ are then given by
a linear combination of the old operators,
\begin{equation}
    D^A_i = \sum_j u^*_{ij} A_j, \quad D^B_i = \sum_j v_{ij} B_j \,,
\end{equation}
which satisfy the same orthogonality condition, $\tr D^A_i D^A_j =
N_A \delta_{ij}$ (and equivalently for $D^B_i$). So we can rewrite
Eq.~\eqref{opbloch} as
\begin{align} \label{opblochsing}
    O = & \, e \, \id_{\rm{d}} \otimes
        \id_{\rm{d}}
        + \sum_{i=0}^{d^2 - 1} \, r_i \, D^A_i \otimes \id_{\rm{d}}
        + \sum_{j=0}^{d^2 - 1} \, t_j \,\id_{\rm{d}} \otimes D^B_j
        + \sum_{i,j} \, s_{i} D^A_i \otimes D^B_i \,,
        \nonumber\\
\end{align}
where $r_i = \sum_j a_j u_{ij}, \ t_j = \sum_k b_k v^*_{jk}$ and
$s_i = \sum_{j,l} u_{ij} c_{jl} v^*_{il}$. We call the decomposed
operator that is written in the optimized way of
Eq.~\eqref{opblochsing} ``singular value optimized'' (SVO). Of
course a product state can then also be decomposed in terms of the
new basis,
\begin{align} \label{prodblochsing}
    \sigma_p = \ & \frac{1}{d^2} \Big( \id_{\rm{d}} \otimes
        \id_{\rm{d}}
        + \sum_{i=0}^{d^2 - 1} f_A \, \bar{n}_i \, D^A_i \otimes \id_{\rm{d}}
        + \sum_{j=0}^{d^2 - 1} f_B \, \bar{m}_j \,\id_{\rm{d}} \otimes D^B_j
        + \sum_{i,j} f_A f_B \, \bar{n}_i \bar{m}_j D^A_i \otimes D^B_j \Big) \,,
        \nonumber\\
    & \bar{n}_{i}, \bar{m}_{j} \in \mathbbm{C}\,, \;\quad \left| \vec{\bar{n}}
        \right| \leq
        1\,,\quad \left| \vec{\bar{m}} \right| \leq 1 \,.
\end{align}

For our purposes we want to reformulate the separability condition
$\eqref{defentwitsep}$ of the EWC:
\begin{cor} \label{corewprod}
    An operator $C \in {\cal A}$ satisfies $\tr \sigma C \geq 0 \ \forall \sigma
    \in S$ if and only if $\tr \sigma_p C \geq 0$ for all pure product
    states $\sigma_p := \rho_A \otimes \rho_B$.
\end{cor}
\begin{proof}
    If we have $\tr \sigma C \geq 0 \ \forall \sigma \in S$ then of
    course also $\tr \sigma_p C \geq 0$ since the pure pruduct states
    $\sigma_p$ are separable states as well. A separable state \eqref{sepstate}
    can be written as a convex combination of \emph{pure} product
    states, $\sigma = \sum_i p_i \sigma_p^i$, because mixed states $\rho_A$ and $\rho_B$ are
    convex combinations of pure states. Thus if $\tr \sigma_p C
    \geq 0$, it follows that $\tr \sigma C = \tr \sum_i p_i \sigma_p^i C  = \sum_i p_i
    \tr \sigma_p^i C \geq 0$ since $p_i \geq 0$.
\end{proof}
At first sight Corollary~\ref{corewprod} may appear redundant, but
it bears the advantage that in order to check if a given operator
satisfies the separability condition \eqref{defentwitsep}, we do not
have to check all separable states but consider pure product states
only, which implies a decrease in effort. Purity of the states is
not essential, but is more convenient in parameterizations.

For an arbitrary operator basis we use the Bloch decomposition
\eqref{opbloch} to write an Hermitian operator $C \in {\cal A}$ as
\begin{align} \label{cbloch}
    C = & \, \delta \left( \mu \id_{\rm{d_A}} \otimes
        \id_{\rm{d_B}}
        + \sum_{i=0}^{d^2_A - 1} \, \tilde{a}_i \, A_i \otimes \id_{\rm{d_B}}
        + \sum_{j=0}^{d^2_B - 1} \, \tilde{b}_j \,\id_{\rm{d_A}} \otimes B_j
        + \sum_{i,j} \, \tilde{c}_{ij} A_i \otimes B_j \right) \,,
        \nonumber\\
        & \mu := \sqrt{(d_A-1)(d_B-1)}, \ \delta \in \mathbbm{R}^+ \,.
\end{align}
Note that we are only interested in Hermitian operators $C$ with positive values of $\delta$. If $\delta$ happens to be negative in the first place we switch to the operator with an additional overall minus sign. For $d_A = d_B = d$ we use Eq.~\eqref{opblochsing} to obtain the SVO
form
\begin{align} \label{cblochsing}
    C = & \, \delta \left( (d-1) \id_{\rm{d}} \otimes
        \id_{\rm{d}}
        + \sum_{i=0}^{d^2 - 1} \, \tilde{r}_i \, D^A_i \otimes \id_{\rm{d}}
        + \sum_{j=0}^{d^2 - 1} \, \tilde{t}_j \,\id_{\rm{d}} \otimes D^B_j
        + \sum_i \, \tilde{s}_{i} D^A_i \otimes D^B_i \right) \,,
        \nonumber\\
\end{align}
and get for the expectation value with product states $\sigma_p$
from Eqs.~\eqref{prodbloch} and \eqref{cbloch} (we use
$\sigma_p^\dag$ in order to conveniently utilize the orthogonality
condition \eqref{mbasisorth})
\begin{equation} \label{expprod}
    \tr \sigma_p^\dag C = \delta \mu \left( 1 + \sqrt{
    \frac{d_B}{N_B(d_B-1)}} \sum_i \tilde{a}_i n^*_i + \sqrt{
    \frac{d_A}{N_A(d_A-1)}} \sum_j \tilde{b}_j m^*_j + \sum_{i,j} \tilde{c}_{ij}
    n^*_i m^*_j \right) \,,
\end{equation}
which simplifies for $d_A = d_B = d$ to (using
Eqs.~\eqref{prodblochsing} and \eqref{cblochsing})
\begin{equation} \label{expprodsing}
    \tr \sigma_p^\dag C = \delta (d-1) \left( 1 + \sqrt{
    \frac{d}{N_B(d-1)}} \sum_i \tilde{r}_i \bar{n}^*_i + \sqrt{
    \frac{d}{N_A(d-1)}} \sum_j \tilde{t}_j \bar{m}^*_j + \sum_i \tilde{s}_i
    \bar{n}^*_i \bar{m}^*_i \right) \,.
\end{equation}
Using the above expressions for the expectation values we obtain a
condition for $\tr C \sigma_p \geq 0$ in Corollary \ref{corewprod}
in terms of Bloch decompositions:
\begin{cor} \label{corbvew}
    Given a decomposition \eqref{cbloch} of an operator $C$ into an arbitrary
    operator basis,
    the expectation value for any product state \eqref{prodbloch} is positive or
    vanishes, $\tr C \sigma_p \geq 0$, if and only if
    \begin{align} \label{corbvewcond}
        & S := \sqrt{
        \frac{d_B}{N_B(d_B-1)}} \sum_i \tilde{a}_i n^*_i + \sqrt{
        \frac{d_A}{N_A(d_A-1)}} \sum_j \tilde{b}_j m^*_j + \sum_{i,j} \tilde{c}_{ij}
        n^*_i m^*_j \geq -1
    \end{align}
    for all Bloch vectors $\vec{n}$, $\vec{m}$. For equal dimensions of the
    subsystems, $d_A = d_B = d$,
    the condition \eqref{corbvewcond} can be simplified to
    \begin{equation} \label{corbvewcondsing}
        S = \sqrt{
    \frac{d}{N_B(d-1)}} \sum_i \tilde{r}_i \bar{n}^*_i + \sqrt{
    \frac{d}{N_A(d-1)}} \sum_j \tilde{t}_j \bar{m}^*_j + \sum_i \tilde{s}_i
    \bar{n}^*_i \bar{m}^*_i \geq -1 \,,
    \end{equation}
    where we used the SVO form \eqref{cblochsing} of $C$.
\end{cor}
\begin{proof}
    The proof is evident from the expressions for the expectation
    values in Eqs.~\eqref{expprod} and \eqref{expprodsing}.
\end{proof}
\begin{rem} \label{remcorbvew}
    Consider the case when there also exists at least one state $\rho$ for
    which $\tr C \rho < 0$. Then $C$ is an entanglement witness if $S \geq -1$.
    Note that by stating
    ``Bloch vector'' we mean vectors $\vec{n}$ that correspond to states
    (i.e. $\rho_A \geq 0$ in Eq.~\eqref{blochonepart}). For arbitrary dimensions
    $d$ of the Hilbert space an arbitrary vector $\vec{n} \in
    \mathbbm{C}^{d^2-1}$ for which $\rho$ has real eigenvalues
    does not always implicate $\rho \geq 0$,
    this is only true for $d=2$, where the familiar matrix basis out
    of the Pauli matrices or rotations thereof is used.
\end{rem}
\begin{rem}
    It directly follows from Corollaries~\ref{corewprod} and
    \ref{corbvew} that $C$ is an optimal entanglement witness if and
    only if there exists a state $\rho$ such that $\tr C \rho < 0$
    and Bloch vectors $\vec{n},\vec{m}$ such that $S = -1$.
\end{rem}
\begin{rem} \label{remmaxmixed}
    For operators $C$ $\eqref{cbloch}$ with vanishing coefficients $\tilde{a}_i$,
    $\tilde{b}_i$, condition \eqref{corbvewcond} reduces to
    \begin{equation} \label{bvcondmaxmix}
        \sum_{i,j} \tilde{c}_{ij}
        n^*_i m^*_j \geq -1
    \end{equation}
    and for $d_A=d_B$ condition \eqref{corbvewcondsing} reduces to
    \begin{equation} \label{bvcondmaxmixsing}
        \sum_i \tilde{s}_i
        \bar{n}^*_i \bar{m}^*_i \geq -1 \,.
    \end{equation}
    This is for example the case
    if we consider geometric operators (see Sec.~\ref{secgew})
    constructed of states that are \emph{locally maximally mixed},
    which means their reduced density matrices are the maximally
    mixed states $(1/d_A) \id$ and $(1/d_B) \id$.
\end{rem}
\begin{lem} \label{lemold}
    For operators $C$ $\eqref{cbloch}$ on a Hilbert space ${\cal H}$ of equal
    dimensional subsystems, $d_A = d_B$, with vanishing coefficients
    $\tilde{a}_i$, $\tilde{b}_i$, the expectation value for product
    states is greater or equal to zero, $\tr \sigma_p \geq 0$, if the
    singular values $\tilde{s}_i$ of the correlation coefficient matrix $\tilde{c}_{ij}$ are smaller
    or equal to one, $\tilde{s}_i \leq 1$.
\end{lem}
\begin{proof}
    With vanishing coefficients $\tilde{a}_i$, $\tilde{b}_i$, the term $S$ in
    Eq.~\eqref{corbvewcond} reduces to $S = \sum_{i,j} \tilde{c}_{ij}
    n_i m_j$. For $d_A = d_B$ we can write the operator in SVO form,
    which gives $S = \sum_i \tilde{s}_i \bar{n}^*_i \bar{m}^*_i$. With
    the condition $s_i \leq 1$ we get
    \begin{equation}
        |S| = \left| \sum_i \tilde{s}_i \bar{n}^*_i \bar{m}^*_i
        \right| \leq \sum_i \tilde{s}_i |\bar{n}^*_i| |\bar{m}^*_i| \leq \sum_i |\bar{n}^*_i|
        |\bar{m}^*_i| \leq 1
    \end{equation}
    and thus $S \geq -1$.
\end{proof}
\begin{rem} \label{remlemold}
    Note that Lemma~\ref{lemold} gives only a sufficient condition
    for satisfying the inequality $\tr \sigma_p \geq 0$. It is
    necessary for dimensions $2 \times 2$ only, since in this case
    any vectors $\vec{n}_i$ and $\vec{m}_i$
    \eqref{prodbloch} correspond to states, see Remark~\ref{remcorbvew},
    and with at least one singular value $s_i \geq 1$ one can easily construct
    Bloch vectors such that $S < -1$. For higher dimensions it is
    possible that some $s_i > 1$, and still there exists no Bloch
    vectors $\vec{n}_i$ and $\vec{m}_i$ (that provide
    $\rho \geq 0$) such that $S < -1$.
\end{rem}

\section{Geometric entanglement witnesses} \label{secgew}

\begin{defin} \label{defgeomop}
    A \emph{geometric operator} $G \in {\cal A}$ is defined as
    \begin{equation} \label{geomop}
        G := \rho_1 - \rho_2 - \langle \rho_1, \rho_1 - \rho_2 \rangle
        \id_D \,,
    \end{equation}
    where $\rho_1$ and $\rho_2$ are arbitrary
    states in ${\cal A}$ and $\rho_1 \neq \rho_2$.
\end{defin}
The definition originates from the construction of entanglement witnesses in Refs.~\cite{pittenger02, bertlmann02}, with the difference that in our definition the geometric operator \eqref{geomop} does not yet have to be an entanglement witness. The construction \eqref{geomop} provides $\tr \rho_1 G = 0$ and $\tr
\rho_2 G < 0$,
\begin{align}
    & \tr \rho_1 G = \langle \rho_1, G \rangle = \langle \rho_1 - \rho_1, \rho_1 - \rho_2 \rangle
    = 0 \,, \nonumber\\
    & \tr \rho_2 G = \langle \rho_2, G \rangle = \langle \rho_2 - \rho_1, \rho_1 - \rho_2 \rangle
    = - \| \rho_1 - \rho_2 \|^2 < 0 \,.
\end{align}
It corresponds to a hyperplane in the Hilbert-Schmidt space ${\cal
A}$ that divides the whole state space into states $\rho_n$ for
which $\tr \rho_n G < 0$ and states $\rho_p$ for which $\tr
\rho_p G \geq 0$, see Ref~\cite{bertlmann05}. The hyperplane
is orthogonal to $\rho_1 - \rho_2$ since for all states $\rho_G$ on
the plane, i.e. that satisfy $\tr \rho_G G = 0$, the operator
$\rho_1 - \rho_2$ is orthogonal to $\rho_G - \rho_1$ because $\tr
\rho_G G = \langle \rho_G - \rho_1, \rho_1 - \rho_2 \rangle = 0$.
\begin{defin} \label{geomentwit}
    A \emph{geometric entanglement witness (GEW)} $A_G$ is a geometric operator
    that satisfies $Tr \sigma_p A_G \geq 0$ for all pure product
    states $\sigma_p$.
\end{defin}
Due to its construction, a geometric entanglement witness (see also
Refs.~\cite{pittenger02, bertlmann02, pittenger03, bertlmann05,
bertlmann08, bertlmann08b} has to witness at least the entanglement
of $\rho_2$. For arbitrary states $\rho_2$ it is easy to construct
geometric operators $G$ (Definition~\ref{defgeomop}) that ensure
$\tr \rho_2 G < 0$, but difficult to confirm that also $\tr \sigma_p G \geq 0$ for all pure product states, which would yield $G = A_G$.
Nevertheless, due to their simple geometric construction, geometric
operators provide useful tools to characterize entanglement, as we
will see in the further sections. Other methods to construct and
optimize entanglement witnesses are given in
Refs.~\cite{lewenstein00, hyllus05, chruscinski08, ioannou04,
ioannou06}.

To detect entanglement it is sufficient to consider geometric
entanglement witnesses only:
\begin{lem}
    Any entangled state is witnessed by a geometric entanglement
    witness.
\end{lem}
\begin{proof}
    If $\rho$ is entangled, then there exists a so-called
    \emph{nearest separable state} $\sigma_0$, i.e. the separable state for
    which the Hilbert-Schmidt distance \eqref{hsdistance} from
    $\rho$ to the set of separable states
    ${\cal S}$ is minimal, because ${\cal S}$ is convex and compact. The corresponding geometric operator
    $\sigma_0 - \rho - \langle \sigma_0, \sigma_0 - \rho \rangle
    \id_D$ is an entanglement witness, since the
    corresponding hyperplane includes $\sigma_0$, is orthogonal to $\sigma_0 - \rho$
    and is therefore tangent to ${\cal S}$. For more details on
    nearest separable states see Refs.~\cite{witte99, bertlmann02,
    bertlmann05}.
\end{proof}
Geometric entanglement witnesses bear the advantage that they can be
``shifted'' along lines of parameterized states.
\begin{prop}[Shift method] \label{shiftmethod}
    If a geometric operator
    \begin{equation} \label{shiftop}
        G_\lambda = \rho_\lambda - \rho - \langle \rho_\lambda,
        \rho_\lambda - \rho \rangle \id_{\rm{D}}
    \end{equation} \label{paramstates}
    with a parameterized family of states
    \begin{equation} \label{rholambda}
        \rho_\lambda := \lambda \rho + (1-\lambda) \tilde{\rho}, \quad 0
        \leq \lambda < 1, \ \rho, \tilde{\rho} \in {\cal A}
    \end{equation}
    is an entanglement witness in a parameter region $\lambda \in [
    \lambda_i, 1 )$, i.e. if it satisfies $\tr \sigma_P G_\lambda \geq 0$
    for all pure product states $\sigma_P$, then $\rho_\lambda$
    is entangled for $\lambda \in ( \lambda_i, 1 ]$.
\end{prop}
\begin{proof}
    We consider states $\rho_\lambda$ with $\lambda_i < \lambda \leq
    1$ and the geometric entanglement witness $A_{\lambda_i} = \rho_{\lambda_i} - \rho
    - \langle \rho_{\lambda_i}, \rho_{\lambda_i} - \rho \rangle
    \id_{\rm{D}}$. The expectation value in $\rho_\lambda$ is
    \begin{align}
        \tr \rho_\lambda A_{\lambda_i} & = \langle \rho_\lambda, A_{\lambda_i} \rangle =
        \langle \rho_\lambda - \rho_{\lambda_i}, \rho_{\lambda_i} - \rho
        \rangle \nonumber\\
        & = (\lambda_i - \lambda) (1- \lambda_i) \langle \rho
        - \tilde{\rho}, \rho - \tilde{\rho} \rangle = (\lambda_i - \lambda) (1- \lambda_i)
        \| \rho - \tilde{\rho} \|^2 < 0 \,,
    \end{align}
    hence the states $\rho_\lambda$ with $\lambda_i < \lambda \leq
    1$ are entangled.
\end{proof}
An effective way to use the shift method of
Proposition~\ref{shiftmethod} is to identify $\tilde{\rho}$ in
Eq.~\eqref{rholambda} with a separable state, and $\rho$ with a
state which is known to be entangled. There are two cases where this
is of particular interest:
\begin{enumerate}
\item (\emph{Outside-in shift.}) Starting from the entangled state
    $\rho$, we can detect further entangled states along the line in
    direction to the separable state. By proofing that $G_\lambda$
    is an entanglement witness for a
    parameter region $\lambda_i \leq \lambda < 1$ (the case $\lambda =
    1$ can be included with a suitable normalization of $G_\lambda$),
    one can infer that all states $\rho_\lambda$ within this region are
    entangled. A reasonable choice for the
    separable state is the maximally mixed state $(1/D) \id$. In this way
    one can detect bound entangled states, for example if we choose
    a PPT entangled ``starting state'' $\rho$, then we are likely to find
    more bound entangled states along the parameterized line $\rho_\lambda$. The
    outside-in shift is illustrated in Fig.~\ref{figoutsidein}. See
    Refs.~\cite{bertlmann08, bertlmann08b} for application examples.
    In Ref.~\cite{bandyopadhyay08} a similar approach with
    parameterized lines between PPT entangled states and the maximally mixed state
    is used to identify families of bound
    entangled states in the context of robustness of entanglement.
    \begin{figure}
        \includegraphics[width=0.4\textwidth]{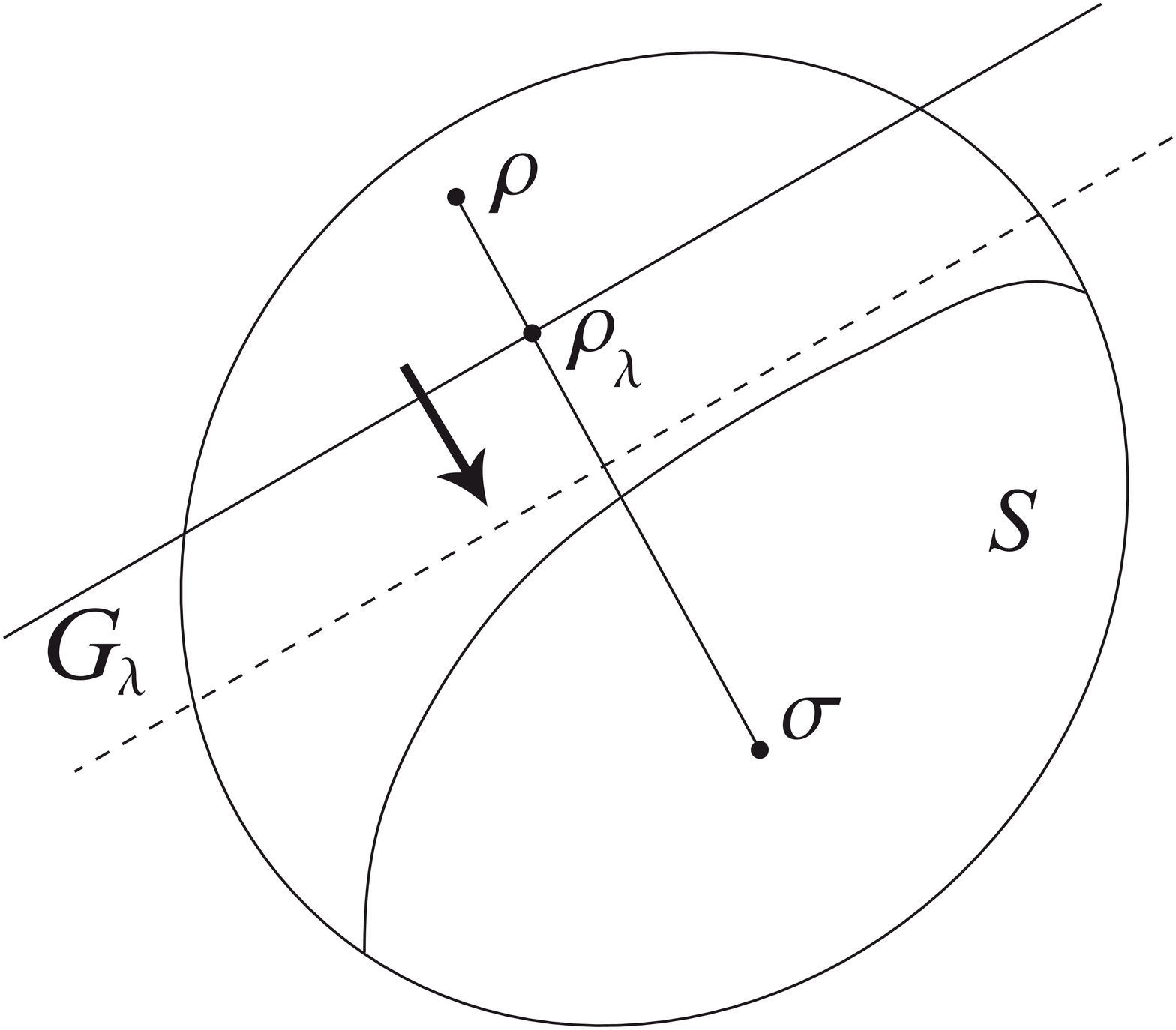}\\
        \caption{Outside-in shift method. On the line between the entangled
        state $\rho$ and the separable state $\sigma$ one can detect more entangled
        states, e.g. bound entangled states, by shifting the geometric entanglement
        witness.} \label{figoutsidein}
    \end{figure}
\item (\emph{Inside-out shift.}) Another application of the shift
    method is the step-by-step construction of the convex set of separable states.
    Here one has to use optimal GEWs that correspond to hyperplanes
    tangent to the set of separable states. Let us assume we are given a specific
    convex subset of states for which we want to determine the entanglement properties
    and that some separable states are
    known. From these we can constuct a \emph{kernel polytope} of
    separable states, i.e. the convex hull of the known separable
    states. Then we assign geometric operators to hyperplanes tangent to
    the kernel polytope. For example, an operator corresponding to a plane
    that includes the line between two separable states
    can be constructed in the following way:
    Given two separable states $\sigma_1$ and $\sigma_2$, the convex
    line between them is $\sigma_{\mu} = \mu \sigma_1 + (1-\mu)
    \sigma_2$. Now we choose an entangled state $\omega$, that of course lies outside the
    kernel polytope, such that there exists a $\mu_i$ with $\tilde{\sigma} =
    \sigma_{\mu_i}$ for which we have the orthogonality condition
    $\langle \sigma_1 - \sigma_2, \omega -
    \tilde{\sigma} \rangle = 0$. The geometric operator is then
    given by
    \begin{equation}
        G = \tilde{\sigma} - \omega - \langle \tilde{\sigma},
        \tilde{\sigma} - \omega \rangle \id_{\rm{D}}
    \end{equation}
    and a shift operator $G_\lambda$ between $\tilde{\sigma}$ and
    $\omega$ according to
    Eqs.~\eqref{shiftop} and \eqref{rholambda}. The construction
    of operators that correspond to boundary planes of the kernel
    polygon identified by more than two separable states is done
    similarily, using more orthogonality conditions and the convex
    hull between three states.

    Once we assigned geometric operators to the boundary of the
    kernel polytope, we utilize Proposition~\ref{shiftmethod} to
    ``shift'' the operators outside and survey the minimum of $S$ in
    Eq.~\eqref{corbvewcond} or \eqref{corbvewcondsing}. At one point of the
    parameterized line \eqref{rholambda} we obtain $S = -1$;
    the geometric operators become optimal geometric
    entanglement witnesses. In this way we can assemble the
    shape of the set of separable states for the considered set of
    states and distinguish it from the set of entangled states. It may likely be
    that we have an idea of the shape of the set of separable states that we got from
    applying necessary separability criteria. Then we can use the inside-out shift
    to verify or falsify that shape: The inside-out shifted geometric operators
    should correspond to optimal GEWs when
    they become tangent to the estimated shape. In this way we get vertices of a new
    polytope, whose boundary planes are shifted again. Thus
    we either verify the estimated shape of separable states, or,
    if a shifted plane is an optimal GEW before it is tangent to the shape,
    it is an enclosure of all
    separable states and also entangled ones. We require finite steps of this method
    if the estimated shape is a polygon, and (in principle) infinite steps if it
    is not a polygon, i.e. if it has a curved surface.

    If we have no idea of a possible shape of the set of
    separable states, or if our estimation turned out to be wrong, we can
    use the inside-out shift to obtain at least a tight enclosure polytope.
    It is a polytope that encloses all separable states but might also contain some
    entangled states, it can be obtained by applying the shift to more
    than one kernel polytope. Both situations are sketched in
    Fig.~\ref{figinsideout}.
    \begin{figure}
        \includegraphics[width=0.6\textwidth]{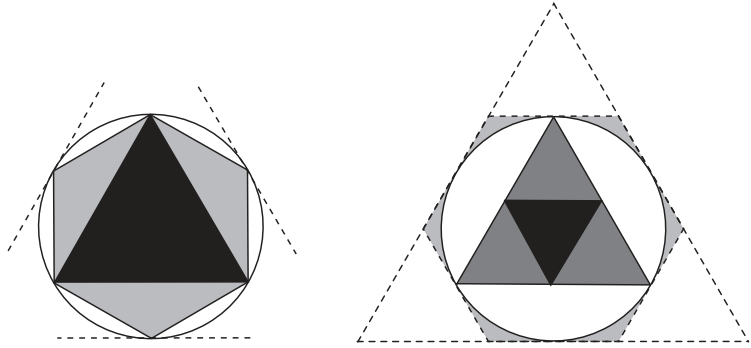}\\
        \caption{Left: Sketch of the inside-out shift with an
        estimate of the shape of separable states, which in this
        case coincides with the true set of separable states,
        pictured by the circle.
        We start with a kernel polytope (black triangle) and shift the boundary planes outside
        until they become optimal GEWs, which are tangents to the circle (dashed lines).
        In this way we can draw a new polytope (hexagon, grey).
        In the next steps (not illustrated) the boundaries of the
        new polytope are shifted and we gain a new polytope, and so
        on. In this way we reconstruct the circle shape.\\
        Right: Sketch of the inside-out shift where we do not rely on an estimate of the shape of
        separable states. The true set of separable states is again pictured by a circle.
        Here we get a first enclosure polytope (biggest triangle with dashed lines),
        by shifting the boundaries of a first kernel polytope outside (dark grey triangle).
        A tighter enclosure polytope (hexagon with dashed lines) is obtained by shifting
        the boundaries of a second kernel polytope (small black triangle) outside. The light
        grey areas mark states inside the enclosure polytope that are not separable and
        thus account for the deviation of the enclosure polytope from the true
        set of separable states.}
        \label{figinsideout}
    \end{figure}
\end{enumerate}
The difficult part of Proposition~\ref{shiftmethod} is to prove that
$G_\lambda$ is an entanglement witness, in particular the
verification of the separability condition \eqref{defentwitsep}. To
accomplish this we can efficiently use the previous corollaries and
lemmas, which will be demonstrated by the example of the next
section.

\section{Entanglement properties of a family of two-qutrit states}
\label{secexample}

An interesting set of states is the \emph{magic simplex} of
two-qudit states (dimension $d \times d$) \cite{baumgartner06,
baumgartner07, baumgartner08}. It is the set of all states that are
mixtures of Bell states $P_{nm}$,
\begin{equation}
    {\cal W} \,:=\, \left\{ \sum_{n,\,m=0}^{d-1} q_{nm} P_{nm} \ | \ q_{nm}
    \geq 0, \ \sum_{n,\,m} q_{nm}=1 \right\},
\end{equation}
where the $d^2$ operators (the \emph{Bell states})
\begin{equation} \label{bellstates}
    P_{nm} \,:=\, (U_{nm} \otimes \mathbbm{1}) | \phi^+_d \rangle
    \langle \phi^+_d | (U_{nm}^\dag \otimes \mathbbm{1})
\end{equation}
form an orthogonal basis of the $d \times d$ dimensional Hilbert
space and the vector state $| \phi^+_d \rangle$ denotes the
maximally entangled state
\begin{equation}
    \left| \phi^+_d \right\rangle = \frac{1}{\sqrt{d}} \,\sum_{j=0}^{d-1} \left| j
    \right\rangle \otimes \left| j \right\rangle.
\end{equation}
The unitary operatos $U_{nm}$ are the Weyl operators
\begin{equation} \label{weylop}
    U_{nm} = \sum_{k=0}^{d - 1} e^{\frac{2 \pi i}{d}\,kn} \,| k \rangle
    \langle (k+m) \,\textrm{mod}\,d| \,,
\end{equation}
which have been introduced in classical theories of discrete phase
space and appear in quantum information theory in the context of
quantum teleportation \cite{bennett93}. In the teleportation
protocol the Bell state basis \eqref{bellstates} is the higher
dimensional generalization of the two-qubit basis and the Weyl
operators $U_{nm}$ are the analogue of the Pauli operators, they
correspond to the operators Bob has to apply in order to obtain the
teleported state. The reduced density operators of states that are
elements of the magic simplex are maximally mixed, but not all
two-qudit states with maximally mixed reduced density operators are
elements of the magic simplex, apart from dimension $2 \times 2$,
where all locally maximally mixed states are included in the
tetrahedron of all Bell state mixtures \cite{baumgartner08}.
Furthermore the magic simplex has a high symmetry in the phase space
of the coefficients $nm$, for a detailed discussion see
Refs.~\cite{baumgartner06, baumgartner07}.

The Weyl operators \eqref{weylop} form an orthogonal operator basis,
\begin{equation} \label{orthonormwo}
\text{Tr} \,U_{nm}^{\dag} U_{lj} = d \,\delta_{nl} \,\delta_{mj}
\end{equation}
and hence can be used for Bloch decompositions. The Bell states
$P_{nm}$ \eqref{bellstates} can be expressed with Weyl operators as
(where the indices have to be taken mod~$d$)
\begin{equation} \label{projbloch}
    P_{nk} = \frac{1}{d^2} \sum_{m,l=0}^{d-1} e^{\frac{2 \pi
    i}{d}(kl-nm)} \,U_{lm} \otimes U_{-lm} = \frac{1}{d^2}
    \sum_{l,m=0}^{d-1} c_{lm} \,U_{lm} \otimes U_{-lm} \,.
\end{equation}
Obviously the Bloch vectors corresponding to the Bell states have a
diagonal but in general complex coefficient matrix $(c_{ij})$, where
$i$ counts the different combinations of $lm$ and $j$ those of $kn$,
and $c_{ii} := c_{lm}$. The singular values of the coefficient
matrix are $s_i = |c_{ii}| = |c_{lm}|$.

Note that a construction of the type \eqref{bellstates} can be done
with any unitary operators that form a matrix basis of the
Hilbert-Schmidt space, obtaining other bases of orthogonal maximally
entangled states.

A subset of the magic simplex of two-qutrit states (dimension $3
\times 3)$ that reveals interesting entanglement characteristics is
the three-parameter family \cite{bertlmann08, bertlmann08b}
\begin{equation} \label{famstates}
    \rho_{\alpha,\beta,\gamma} := \frac{1-\alpha -\beta -\gamma}{9}
    \mathbbm{1} + \alpha P_{00} + \frac{\beta}{2} \left( P_{10} +
    P_{20} \right) + \frac{\gamma}{3} \left( P_{01} +P_{11}+P_{21}
    \right)\,,
\end{equation}
where the parameters are constrained by the positivity requirement
$\rho_{\alpha,\beta,\gamma} \geq 0$,
\begin{alignat}{2} \label{famstatespos}
    \alpha &\leq \frac{7}{2} \beta +1 -\gamma \,, \quad
    & \alpha &\leq -\beta +1 -\gamma \,, \nonumber\\
    \alpha &\leq -\beta +1 +2\gamma \,, \quad
    & \alpha &\geq \frac{\beta}{8} - \frac{1}{8} + \frac{1}{8}\gamma \,.
\end{alignat}
The family of states \eqref{famstates} contains a one-parameter
family of states that have three entanglement properties; they can
be separable, PPT entangled and NPT entangled. We call them
Horodecki states $\rho_b$ \cite{horodecki99c},
\begin{equation} \label{horstates}
    \rho_b = \frac{2}{7} \left| \phi_+^3 \right\rangle \left\langle
    \phi_+^3 \right| \,+\, \frac{b}{7} \, \sigma_+ \,+\, \frac{5 - b}{7} \,
    \sigma_- \,, \qquad 0 \leq b \leq 5 \,,
\end{equation}
and, according to our parametrization,
\begin{equation} \label{horstatessimplex}
    \rho_b := \rho_{\alpha, \beta, \gamma} \qquad\mbox{with}\quad \alpha =
    \frac{6-b}{21}, \;\beta = -\frac{2b}{21}, \;\gamma = \frac{5-2b}{7}
    \,.
\end{equation}
Using the PPT criterion we find regions of PPT and NPT Horodecki
states: They are NPT for $0 \leq b < 1$, PPT for $1 \leq b \leq 4$
and again NPT for $4 < b \leq 5\,$. In Ref.~\cite{horodecki99c} it
is shown that the states are separable for $2 \leq b \leq 3$ and
bound entangled for $3 < b \leq 4$.

Now let us apply the PPT criterion \eqref{PPT} and the realignment
criterion \eqref{realign} to our three-parameter family
\eqref{famstates}. The PPT criterion provides the following
parameter constraints for PPT states $\rho_{\alpha,\beta,\gamma}$:
\begin{eqnarray} \label{famstatesPPT}
    \alpha &\leq& - \beta - \frac{1}{2} + \frac{1}{2}\gamma \,, \nonumber\\
    \alpha &\leq& \frac{1}{16} \left( -2 + 11\beta + 3 \sqrt{\Delta}\right)\,, \quad
    \alpha \geq \frac{1}{16} \left( -2 + 11\beta - 3 \sqrt{\Delta}\right) \,,
\end{eqnarray}
where $\Delta = 4 + 9 \beta^2 + 4\gamma - 7\gamma^2 - 6\beta (2 +
\gamma)$. Hence all states $\rho_{\alpha,\beta,\gamma}$ with
constraints \eqref{famstatesPPT} are either bound entangled our
separable, whereas the others are NPT entangled.

From the realignment criterion we obtain the constraints
\begin{align}
    \alpha \leq & \ \frac{1}{16} \left(6+11 \beta -\gamma
        - \Delta_1 \right) \label{realignconstraint} \\
    \alpha \leq & \ \frac{1}{16} \left(6+11 \beta -\gamma
        + \Delta_1 \right) \\
    \alpha \geq & \ \frac{1}{16} \left(-6+11 \beta -\gamma
        - \Delta_2 \right) \\
    \alpha \geq & \ \frac{1}{16} \left(-6+11 \beta -\gamma
        + \Delta_2 \right)
\end{align}
where
\begin{align}
    \Delta_1 := & \ \sqrt{4+36 \beta +81 \beta ^2-12 \gamma -54 \beta \gamma
        +33 \gamma ^2} \quad \mbox{and} \nonumber\\
    \Delta_2 := & \ \sqrt{4-36 \beta +81 \beta ^2+12
        \gamma -54 \beta  \gamma +33 \gamma ^2}.
\end{align}
Only constraint \eqref{realignconstraint} is violated by some PPT
states, which thus have to be bound entangled. The PPT entangled
states exposed by the realignment criterion are therefore
concentrated in the region confined by the constraints
\begin{equation} \label{famstatesbe}
\alpha \leq \frac{7}{2} \beta +1 -\gamma, \ \alpha \leq \frac{1}{16}
\left( -2 + 11\beta + 3 \sqrt{\Delta}\right), \ \alpha \geq \
\frac{1}{16} \left(6+11 \beta -\gamma - \Delta_1 \right) \,.
\end{equation}
The three-parameter family \eqref{famstates} also bears the
advantage that it can be nicely illustrated by the Euclidean geometry.
To do this, note that the orthogonality conditions of the
Hilbert-Schmidt space ${\cal A}$ have to be transferred correctly,
which is achieved by choosing a nonorthogonal and differently scaled
coordinate system of parameter axes $\alpha$, $\beta$, and $\gamma$.
They are chosen such that they each become orthogonal to one of the
boundary planes of the set of the three-parameter family of states,
given by the positivity constraints \eqref{famstatespos}. In order
to calculate quantities well known in an Euclidean space spanned by
an orthogonal equally scaled coordinate system, we have to transform
points of the non-orthogonal coordinates $(\alpha, \beta, \gamma)$
into points of orthogonal coordinates $(a, b, c)$ and vice versa by
\begin{equation} \label{orthcoord}
    a = \alpha - \frac{1}{8} \beta - \frac{1}{8} \gamma, \ b =
    \frac{\sqrt{3}}{8} \left( 3 \beta - \gamma \right), \ c =
    \frac{\sqrt{3}}{4} \gamma \,.
\end{equation}
In Fig.~\ref{figallcrit} the three-parameter family of states
$\rho_{\alpha,\beta,\gamma}$ \eqref{famstates} including NPT
entangled, PPT entangled (bound entangled) and further PPT states
are illustrated in the Euclidean geometry picture.
\begin{figure}
    \includegraphics[width=0.7\textwidth]{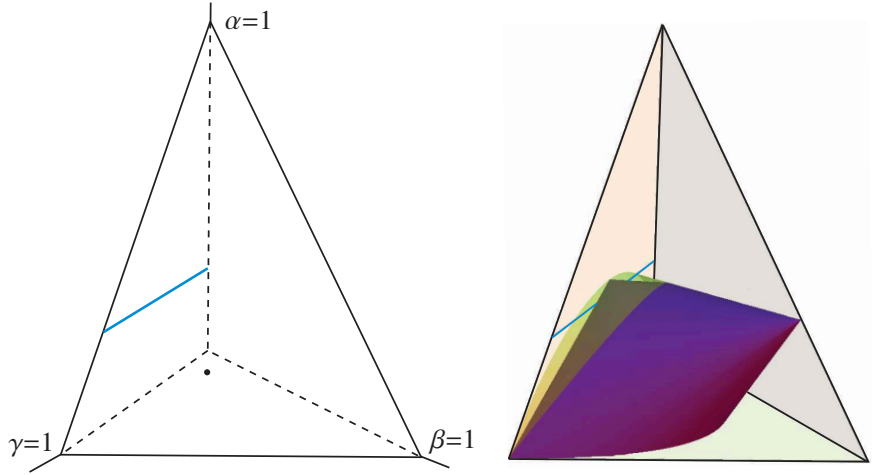}\\
    \caption{Illustration in of the family of states
    $\rho_{\alpha,\beta,\gamma}$ \eqref{famstates} in the Euclidean
    geometry. Left: All states $\rho_{\alpha,\beta,\gamma}$ lie
    within a pyramid due to the positivity constraints
    \eqref{famstatespos}. The dot represents the origin of the
    coordinate axes, which is the maximally mixed state $(1/9) \id$.
    The line (blue) on the left boundary plane represents the
    Horodecki states \eqref{horstates}.
    Right: Illustration of the PPT and realignment criteria. The
    cone with tip on the right vertex line of the pyramid
    contains the PPT states, which is intersected by a cone (tip on
    the left boundary plane of the pyramid) of
    states that satisfy the realignment criterion, and hence PPT
    entangled states can be revealed (translucent yellow region).
    All other states of the pyramid are NPT entangled.} \label{figallcrit}
\end{figure}
In Refs.~\cite{bertlmann08, bertlmann08a} we applied the outside-in
shift method to detect most of the bound entangled states
\eqref{famstatesbe}, where a version of Lemma~\ref{lemold} (for Weyl
operator decompositions) was used to show that for particular
parameter regions geometric operators correspond to geometric
entanglement witnesses. The geometric shifting operators $G_\lambda$
\eqref{shiftop} were constructed on lines between bound entangled
starting states $\rho = \rho_b^{\rm{BE}}$ on the Horodecki line
\eqref{horstates} and the maximally mixed state, $\tilde{\rho} =
(1/9) \id$.

Actually all PPT entangled states of Eq.~\eqref{famstatesbe},
Fig.~\ref{figallcrit}, can be detected using Lemma~\ref{lemold}. To
see this, we construct tangent planes onto the surface of the
function
\begin{equation} \label{realignfunction}
\alpha = \ \frac{1}{16} \left(6+11 \beta -\gamma - \Delta_1 \right)
\end{equation}
from the realignment criterion \eqref{realignconstraint}, where we
use orthogonal coordinates \eqref{orthcoord}. In this way we can
assign geometric operators to the tangential plane by choosing
points $\vec{a}$ inside the planes and points $\vec{b}$ outside the
planes such that $\vec{a}- \vec{b}$ is orthogonal to the planes.
Since the Euclidean geometry of our picture is isomorphic to the
Hilbert-Schmidt geometry, the points $\vec{a}$ and $\vec{b}$
correspond to states $\rho_a$ and $\rho_b$ and we can construct the
geometric operator accordingly,
\begin{equation} \label{geomopre}
G_{\rm{re}} = \rho_a - \rho_b - \langle \rho_a, \rho_a - \rho_b
\rangle \id_9 \,.
\end{equation}

This operators \eqref{geomopre} are linear combinations of the
three-parameter states $\rho_{\alpha, \beta, \gamma}$ which are
linear combinations of the Bell states $P_{nm}$ \eqref{bellstates}
and can be written as a Bloch decomposition using
Eq.~\eqref{projbloch}. First we need to define some expressions of
Weyl operator combinations,
\begin{align} \label{defu1u2}
    U_1 := \ & U_{01} \otimes U_{01} + U_{02} \otimes U_{02} + U_{11}
    \otimes U_{-11} + U_{12} \otimes U_{-12} + U_{21} \otimes U_{-21}
    + U_{22} \otimes U_{-22} \,, \nonumber\\
    U_2 := \ & U^I_2 + U^{II}_2  \qquad\mbox{with}\quad
    U^I_2 := U_{10} \otimes U_{-10} \,,\quad U^{II}_2 := U_{20} \otimes U_{-20}
    \,.
\end{align}
The geometric operators \eqref{geomopre} corresponding to tangent
planes in points $(\alpha_{\rm t}, \beta_{\rm t}, \gamma_{\rm t})$,
where $\alpha_{\rm t}$ is a function of $\beta_{\rm t}$ and
$\gamma_{\rm t}$, given by the realignment function
\eqref{realignfunction}, are
\begin{align} \label{geomopreweyl}
    G_{\rm{re}} = & \ a \,(2 \,\mathbbm{1} - U_1 + c \, U^I_2 \,+\, c^* U^{II}_2 )
        \,, \quad \mbox{with} \nonumber\\
    &a = \frac{1}{36} \left(-2-9 \beta_{\rm t}+3 \gamma_{\rm t}+3
        \Delta_c \right), \nonumber\\
    &c = \frac{9 \gamma_{\rm t}^2+(-2-9 \beta_{\rm t}+3 \gamma_{\rm t})
        \Delta_c + \sqrt{3} \gamma_{\rm t} \left(2+9 \beta_{\rm t}-3 \gamma_{\rm t}+3
        \Delta_c \right) i}{(2+9 \beta_{\rm t})^2-6 (2+9 \beta_{\rm t}) \gamma_{\rm t}
        +36 \gamma_{\rm t}^2} \,, \nonumber\\
    &\Delta_c := \sqrt{4+36 +81 \beta_{\rm t}^2-12 \gamma_{\rm t}-54
        \beta_{\rm t} \gamma_{\rm t}+33 \gamma_{\rm t}^2} \,.
\end{align}
The singular values of the correlation coefficient matrix are the
absolute values of the coefficients $-1$, $c$ and $c^*$ in
Eq.~\eqref{geomopreweyl}, which are all one,
\begin{equation}
    \{ s_i \} = \{ 1, 1, 1, 1, 1, 1, 1, 1 \} \,,
\end{equation}
and therefore, according to Lemma~\ref{lemold}, the geometric
operators $G_{\rm{re}}$ are entanglement witnesses that detect the
entanglement of all states ``above'' the corresponding planes, thus
also the bound entangled states in the region of
Eq.~\eqref{famstatesbe}.

We might ask ourselves if the PPT entanglement of
Eq.~\eqref{famstatesbe}, revealed by the realignment criterion and
also by GEWs, is all there is for the three-parameter states
\eqref{famstates}. Or, to put it differently, are all the
three-parameter states that satisfy
\emph{both} the PPT and the realignment criterion separable? We can
answer this question by using GEWs and the inside-out shift method.
The entanglement properties of the states on the boundary plane
\begin{equation} \label{boundplane}
    \alpha = \frac{7}{2} \beta +1 -\gamma
\end{equation}
of the positivity pyramid are already fixed. The realignment
function \eqref{realignfunction} and also the GEWs $G_{\rm re}$
\eqref{geomopreweyl} draw a triangle on this plane, whose vertices
are separable states. The tip of the triangle is a separable state
since it is PPT and $\gamma = 0$ (all PPT states of the
two-parameter subset $\gamma = 0$ are separable, shown in
Ref.~\cite{baumgartner06}), the other two, at $\gamma = 1$ and
$\gamma = -1$ are simple mixtures of Bell states $P_{nm}$ that are
also shown to be separable in Ref.~\cite{baumgartner06}. So the the
triangle is the convex hull of the three separable states and thus
has to be separable. For an illustration of the entanglement
properties on the boundary plane \eqref{boundplane} see
Fig.~\ref{figpolygonplane}.
\begin{figure}
    \includegraphics[width=0.7\textwidth]{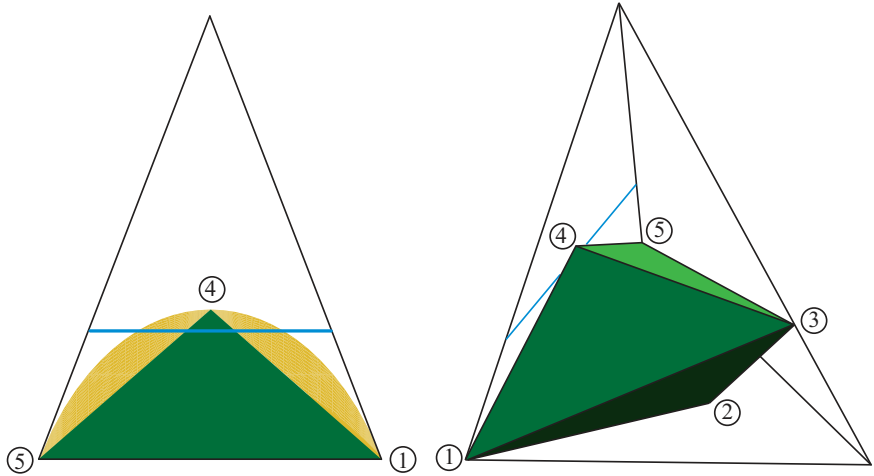}\\
    \caption{Left: The entanglement properties of the
    three-parameter family on the boundary plane \eqref{boundplane} where the Horodecki
    states are located. The triangular region (green) contains the
    separable states, the bound entangled states are located in the
    parabolic region (yellow), and the remaining states are NPT
    entangled. Right: The kernel polytope is a polygon (green) that includes
    states that are necessarily separable.} \label{figpolygonplane}
\end{figure}

But what about all the three-parameter states \eqref{famstates}?
First, we construct a kernel polytope of those states that are
necessarily separable. This can be done by identifying five separable
states that serve as vertices for the kernel polytope. Three arise
from the two-parameter subset $\gamma = 0$, where all PPT states are
separable, the remaining two vertices are the separable states with
$\gamma = 1$ and $\gamma = -1$ on the boundary plane
\eqref{boundplane}. The resulting kernel polytope is a polygon with
five vertices, see Fig.~\ref{figpolygonplane}. Alternatively, one can also use sufficient criteria for separability to construct a kernel polytope of separable states. In Ref.~\cite{chruscinski08a} a sufficient separability criterion is presented that is shown to be applicable for states of the magic simplex.

In Ref.~\cite{bertlmann08b} it remained open if this polygon
contains all separable states of the three-parameter family
\eqref{famstates}, which would imply much greater regions of bound
entanglement than detected before. Here we want to show that this is
not the case.

We can assign geometric operators to four boundary planes of the
kernel polygon, in the same way as we did for the planes on the
realignment surface, see Eq.~\eqref{geomopre}, where we use the
geometric isomorphism again. We call the four geometric operators
$G^u_{\pm}, G^d_{\pm}$, which correspond to the following planes
given by three vertex points (see Fig.~\ref{figpolygonplane}):
$G^u_+$ to $1,3,4$, $G^u_-$ to $3,4,5$, $G^d_+$ to $1,2,3$, and
$G^d_-$ to $2,3,5$. The plus and minus sign indicates the side with
positive or negative values of the parameter $\gamma$. The operators
are
\begin{align} \label{polygonop}
    G^u_\pm = & \ a \,(2 \,\mathbbm{1} - U_1 + c \, U^I_2 \,+\, c^* U^{II}_2 )
        \,, \quad \mbox{with } a = \frac{1}{63}, \ c = -1 \pm
        \sqrt{3}i \nonumber\\
    G^d_\pm = & \ a \,(2 \,\mathbbm{1} + U_1 + c \, U^I_2 \,+\, c^* U^{II}_2 )
        \,, \quad \mbox{with } a = \frac{1}{63}, \ c = -1 \pm
        \sqrt{3}i
\end{align}
The boundary planes can be easily shifted along parameterized lines
through their normal vectors, and so can the assigned geometric
operators \eqref{polygonop}. Note that we have a simplified picture
of locally maximally mixed states, see Remark~\ref{remmaxmixed}. The
operators $G^u_{\pm}, G^d_{\pm}$ themselves are not entanglement
witnesses, since the condition \eqref{bvcondmaxmix} can be
numerically shown to be violated (see Corollary~\ref{corbvew}). The
singular values are again the absolute values of the correlation
coefficients, $\{s_i\} = \{1,1,1,1,1,1,2,2\}$, hence
Lemma~\ref{lemold} does not give an answer. It is difficult to show
a violation analytically because of the complex Bloch vector
geometry of qutrits, see Remarks~\ref{remcorbvew} and
\ref{remlemold}. In order to check the condition \eqref{bvcondmaxmix} we minimize the left-hand term numerically by varying the
possible Bloch vectors $\vec{n}^*, \vec{m}^*$, restricted by the
condition $\rho \geq 0$ with $\rho$ of Eq.~\eqref{blochonepart}.
Shifting the operators outside, we find a minimum $S = -1$ of
condition \eqref{bvcondmaxmix} when the planes become tangent to the
shape enclosed by the PPT and realignment criterion, achieving new
vertices at the touch points. Employing the inside-out shift method,
see Sec.~\ref{secgew} and Fig.~\ref{figinsideout}, we construct a
new polygon with the new vertices, and assign new geometric
operators corresponding to the new boundary planes. Shifting the new
operators outside, we again find the minimum $S = -1$ at planes
tangent to the PPT and realignment shape. Therefore there is a very
strong implication that the PPT and realignment shape, seen as the
two-cone shape in Fig.~\ref{figallcrit}, is the shape of the
separable states. Fig.~\eqref{figallcrit} thus is a picture of all
entanglement properties of the three-parameter family.

\section{Summary and conclusion}

We use the concept of Bloch decompositions and entanglement
witnesses to detect the entanglement properties of arbitrary
dimensional bipartite quantum states. In particular we show how to
reformulate the conditions of the entanglement witness criterion by
using Bloch decompositions (Corollary~\ref{corbvew}) and formulated
a sufficient condition for an operator to be an entanglement witness
(Lemma~\ref{lemold}).

We give the definition of a geometric operator and a geometric
entanglement witness and explain two methods of ``shifting'' it
(Proposition~\ref{shiftmethod}): One for the detection of bound
entangled states, the \emph{outside-in shift}, and one for the
detection of separable states and for the construction of the shape
of the set of separable states, the \emph{inside-out shift}.

Finally we apply the previous results on a family of three-parameter
two-qutrit states that are part of a simplex in the state space of
two qutrits, the \emph{magic simplex}. We show how to detect bound
entangled states and construct the shape of separable states for
this family. The results can be conveniently illustrated by the
Euclidean geometry.

Our approach to entanglement detection is guided by the
geometrically intuitive way of using entanglement witnesses. The
construction of geometric entanglement witnesses directly uses the
fact that entanglement witnesses correspond to hyperplanes in the
Hilbert-Schmidt geometry. In this way it becomes easier to apply
geometric operations like the shifting of planes. Using Bloch vector
decompositions of operators and states we can furthermore simplify
the conditions that have to be satisfied such that a geometric
operator is a geometric entanglement witnesses. The construction of
geometric entanglement witnesses does not rely on special properties
of the states, i.e. it can be done for NPT or PPT entangled states
likewise.

The presented example is relevant in many aspects. First of all the
states of the magic simplex are a higher dimensional analogy of
Bell-state mixtures of the two-qubit case that are relevant for
quantum communication tasks, as explained in Sec.~\ref{secexample}.
Furthermore it is interesting and surprising that this particular
three-parameter family includes the Horodecki states that were among
the first examples of bound entangled states. Thus the
three-parameter family can be viewed as a more-parameter extension
of the Horodecki states that includes even more bound entangled
states. Finally the three-parameter states allow a nice Euclidean
illustration that makes the regions of entangled, bound entangled
and separable states visible.

Throughout the paper we restrict ourselves to bipartite states. Of
course a multipartite extension is trivially possible if we only
want to distinguish between states that contain entangled states in
any of its particles and states that are fully separable into all
particles. The definition of separable states just has to be
extended with additional tensor products respectively. In the case
of multipartite states one can distinguish between the
distillability of states into entangled states of a fixed number of
particles \cite{hiesmayr08a}. Entangled multipartite states can
themselves be classified in different ways, for example with respect
to the number of particles that are entangled. For details see,
e.g., Refs.~\cite{duer00a, acin01, hiesmayr08}.

\begin{acknowledgments}

The author would like to thank Reinhold A. Bertlmann, Beatrix C.
Hiesmayr, and Marcus Huber for helpful discussions. This research
has been financially supported by FWF project CoQuS no. W1210-N16 of
the Austrian Science Foundation.

\end{acknowledgments}

\bibliography{references}

\end{document}